\documentclass[article,onefignum,onetabnum]{siamonline171218}

\usepackage{lipsum}
\usepackage{amsfonts}
\usepackage{graphicx}
\usepackage{epstopdf}
\usepackage{algorithmic}
\ifpdf
  \DeclareGraphicsExtensions{.eps,.pdf,.png,.jpg}
\else
  \DeclareGraphicsExtensions{.eps}
\fi

\usepackage{enumitem}
\setlist[enumerate]{leftmargin=.5in}
\setlist[itemize]{leftmargin=.5in}

\headers{OT Divergences induced by Scoring Functions}{Pesenti and Vanduffel}

\title{Optimal Transport Divergences induced by Scoring Functions\thanks{This version: 10. April 2024; first version: 20. November 2023.
The authors thank Tobias Fissler, Cale Rankin, Ludger Rüschendorf, and Ting-Kam Leonard Wong for stimulating discussions on the topic. \funding{SP gratefully acknowledges the support of the Natural Sciences and Engineering Research Council of Canada (NSERC) with funding reference numbers DGECR-2020-00333 and RGPIN-2020-04289. SV thanks Research Foundation Flanders (FWO) for financial support with funding numbers  FWO SBO S006721N and FWO WOG W001021N.}}}

\author{Silvana M. Pesenti\thanks{Department of Statistical Sciences, University of Toronto, Canada
  (\email{silvana.pesenti@utoronto.ca}).}
\and Steven Vanduffel
\thanks{Department of Economics and Political
Science, Vrije Universiteit Brussel, Belgium
  (\email{Steven.Vanduffel@vub.be}).
  }
}

\usepackage{amsopn}

\makeatletter
\newcommand*{\addFileDependency}[1]{
  \typeout{(#1)}
  \@addtofilelist{#1}
  \IfFileExists{#1}{}{\typeout{No file #1.}}
}
\makeatother

\usepackage{todonotes}
\usepackage{float}
\usepackage{setspace}
\usepackage{amsmath}
\usepackage{enumitem}
\usepackage{amssymb,amsmath}
\usepackage{wrapfig}
\usepackage{dsfont, mathrsfs}

\newsiamremark{assumption}{Assumption}
\newsiamremark{remark}{Remark}
\newsiamremark{hypothesis}{Hypothesis}
\crefname{hypothesis}{Hypothesis}{Hypotheses}
\newsiamthm{claim}{Claim}
\newsiamremark{optimisation}{Optimisation}

\newcommand{\A}{{\mathsf{A}}}

\newcommand{\R}{{\mathds{R}}}
\newcommand{\N}{{\mathds{N}}}
\newcommand{\Id}{{\mathds{1}}}

\newcommand{\E}{{\mathbb{E}}}
\renewcommand{\P}{{\mathbb{P}}}
\newcommand{\Q}{{\mathbb{Q}}}

\newcommand{\F}{{\mathcal{F}}}
\renewcommand{\L}{{\mathcal{L}}}
\newcommand{\M}{{\mathcal{M}}}
\newcommand{\bM}{\bar{{\mathcal{M}}}}
\newcommand{\Mo}{{\mathcal{M}^0}}
\newcommand{\Minv}{{\breve{\mathcal{M}}}}

\newcommand{\bz}{{\boldsymbol{z}}}

\newcommand{\DB}{{\mathscr{B}}}
\newcommand{\DS}{{\mathscr{S}}}

\newcommand{\VaR}{{\textrm{VaR}}}

\newcommand{\ep}{{\varepsilon}}
\renewcommand{\d}{{\mathrm{d}}}

\DeclareMathOperator*{\argmin}{argmin}
\DeclareMathOperator*{\argmax}{argmax}

\renewcommand{\Finv}{{\Breve{F}}}
\newcommand{\Ginv}{{\Breve{G}}}

\newcommand{\Fxiinv}
{{\Breve{F}_{\xi}}}

\hypersetup{
  colorlinks   = true, 
  urlcolor     = blue, 
  linkcolor    = red, 
  citecolor   = blue 
}

\begin{document}

\maketitle
\begin{abstract}
We employ scoring functions, used in statistics for eliciting risk functionals, as cost functions in the Monge-Kantorovich (MK) optimal transport problem. This gives raise to a rich variety of novel asymmetric MK divergences, which subsume the family of Bregman-Wasserstein divergences. We show that for distributions on the real line, the comonotonic coupling is optimal for the majority of the new divergences. Specifically, we derive the optimal coupling of the MK divergences induced by functionals including the mean, generalised quantiles, expectiles, and shortfall measures. Furthermore, we show that while any elicitable law-invariant coherent risk measure gives raise to infinitely many MK divergences, the comonotonic coupling is simultaneously optimal. \\
The novel MK divergences, which can be efficiently calculated, open an array of  applications in robust stochastic optimisation.  We derive sharp bounds on distortion risk measures under a Bregman-Wasserstein divergence constraint, and solve for cost-efficient payoffs under benchmark constraints.  
\end{abstract}

\begin{keywords}
Optimal transport, risk measures, scoring functions, elicitability, coupling, Wasserstein distance, asymmetric Optimal Transport
\end{keywords}


\section{Introduction}

This work connects optimal transport with the statistics of point forecasts and risk measure theory to derive asymmetric optimal transport divergences. Initially developed by Gaspard Monge in the 18-th century to establish the most efficient way of moving piles of soil from one location to another, optimal transport (OT) theory has evolved as a flourishing mathematical discipline that studies the optimal transportation of one distribution to another; see e.g. the monographs \cite {rachev1998mass} and \cite{Santambrogio2015book}. The theory of OT has numerous applications in various distinct disciplines including economics, engineering, and the natural sciences. For instance, economists use OT to model the flow of goods and resources in financial markets. It find itself in causal inference, partial identification, and the analysis of distributional treatment effects in e.g., biostatistics. In image processing and computer vision, OT is used for solving matching and registration problems whereas in the domain of fluid dynamics it optimises the flow of gases.

The most commonly used and best understood OT distance is the (2-)Wasserstein distance, that arises as the minimiser of the Monge-Kantorovich OT problem with the symmetric cost function $c(z_1,z_2) = (z_1-z_2)^2$. In various  problems of interest, however, asymmetry may be desired. A case in point is the quantification of model ambiguity in optimal portfolio strategies, where a decision maker assigns a larger costs to potential losses than to equally large gains (if any). The literature on asymmetric cost function is scarce, with the exception of the recently introduced Bregman-Wasserstein (BW) divergences \cite{carlier2007monge}. A BW divergence is the minimiser of the Monge-Kantorovich OT problem in which the cost function is a Bregman divergence generated by a strictly convex function $\phi(\cdot)$; we refer to \cite{Rankin2023WP} for discussions and information geometric interpretations. The BW divergence is an asymmetric generalisation of the Wasserstein distance, recovered for $\phi(x)=x^2$, that allows the modelling of dissimilarities.

The study of elicitability is a fast growing field in statistics and at its core are scoring functions that incentivise truthful predictions and allow for forecast comparison, model comparison (backtesting), and model calibration \cite{Gneiting2011,Fissler2021EJS}. In sensitivity analysis, scoring functions are utilised for defining sensitivity measures which quantify the sensitivity of an elicitable risk measure to perturbations in the model’s input factors \cite{Fissler2023EJOR}. The most well-known family of scoring functions are the Bregman divergences that elicit the mean, where a functional is called elicitable if it is a minimiser of an expected score, see \Cref{def:elicitable}. Other elicitable functionals are quantiles, expectiles, and shortfall risk measures; tools used in risk management. Scoring functions are by nature asymmetric, making them ideal candidates for asymmetric cost functions in the Monge-Kantorovich OT problem. Indeed, we propose novel asymmetric Monge-Kantorovich (MK) divergences where the OT cost functions are statistical scoring functions. As a Bregman divergence elicits the mean and gives raise to a BW divergence, our new MK divergences can be seen as generalisations of BW divergences, and thus the Wasserstein distance. In addition to scoring functions that elicits the mean, we study scoring functions that elicit the quantile, the expectile, and law-invariant convex risk measures. Interestingly, we find that most of the introduced MK divergences are attained by the comonotonic coupling. Furthermore, as an elicitable functional possesses infinitely many scoring functions, and thus gives raise to infinitely many MK divergences, the comonotonic optimal coupling is typically simultaneously optimal. Using the celebrated Osband's principle in statistics, we propose ways to create novel MK divergences that are attained by the anti- or comonotonic  coupling. Furthermore, we prove that MK divergences induces by any law-invariant elicitable coherent risk measure are attained by the comonotonic coupling. Finally, we provide two applications to robust stochastic optimisation. First, we derive sharp bounds on distortion risk measures when admissible distributions belong to a BW-ball around a reference distribution, thus  significantly generalising recent results of \cite{bernard2022robust}, who solve this problem for the special case of a Wasserstein ball. Second,  we find the cheapest payoff (reflecting terminal wealth) under the constraint that its distribution lies within a BW-ball around a benchmark distribution.

This paper is organised as follows. \Cref{sec:MK-divergence} introduces the MK divergences after reviewing the statistical concepts elicitability and scoring functions and the relevant topics in OT. \Cref{sec:OT-maps} is devoted to MK divergences induced by elicitable risk functionals such as the quantile, expectile, and shortfall risk measure. We find that for distributions on the real line the majority of the new MK divergences are attained by the comonotonic coupling. Applications of the new divergences to risk measure bounds, significantly generalising recent results by \cite{bernard2022robust}, and portfolio management are provided in \Cref{sec:app}.

\section{Monge-Kantorovich divergences induced by scoring functions}\label{sec:MK-divergence}

\subsection{Elicitability and scoring functions}
We first review the statistical concepts of elicitability and scoring functions by following the traditional statistical notation and decision theoretic setup, see e.g., \cite{Gneiting2011} and \cite{Fissler2021EJS}. For this let $(\Omega, \F, \P)$ be a complete probability space and denote by $\L:= \L(\Omega, \F, \R)$ the space of all random variables. The cumulative distribution function (cdf) of a random variable $X \in \L$ is denoted by $F_X(\cdot):= \P(X\le \cdot)$, and we write $\Mo: = \Mo(\R)$ to denote the space of all cdfs on $\R$. For a cdf $F\in\Mo$, we define its corresponding (left-continuous) quantile function by $\Finv(u):= \inf\{y \in \R ~|~ F(y) \ge u\}$, $u \in [0,1]$, with the convention that $\inf\emptyset = + \infty$. Throughout, we will use the notation $\M \subseteq \bM \subseteq \Mo$ to denote sub-classes of cdfs.

\begin{definition}[Scoring function]
A scoring function (or score) is a measurable map $S\colon\A\times\R\to[0,\infty]$, where $\A \subseteq \R$ is called the action domain. For a given functional $T\colon\bM\to \A$ and a sub-class $\M\subseteq \bM$, the scoring function $S$ may satisfy the following properties:
\begin{enumerate}[label = $(\roman*)$]
    \item 
    $S$ is \emph{$\M$-consistent} for $T$, if
    for all $F\in\M$ and for all $z\in\A$
    \begin{equation}
    \label{eq:consistency}
    \int S\big(T(F),y\big)\,\d F(y) \;\le\; \int S(z,y)\,\d F(y)\,.
    \end{equation}

    \item
    $S$ is \emph{strictly} $\M$-consistent for $T$, if it is $\M$-consistent for $T$ and if the inequality in \eqref{eq:consistency} is strict for all $z\neq T(F)$.   
\end{enumerate}
\end{definition}

Throughout, we make the following non-restrictive assumptions on the considered scoring functions.
\begin{assumption}[Normalisation of scores]\label{asm:score}
    Let $S$ be an $\M$-consistent score for $T$ and denote by $\delta_y$, $y \in \R$, point measures. Then it holds that 
    \begin{enumerate}[label = $(\roman*)$]
        \item \label{asm:eq:dirac}
        $S\big(T(\delta_y), y\big) < S(z,y) $ for all $z \neq T(\delta_y)$ and $y \in \R$, and 
        \item \label{asm:eq:normalised}
        $S\big(T(\delta_y), y\big) = 0 $ for all $y\in\R$.        
    \end{enumerate}
\end{assumption}
\cref{asm:score} means that the scores are strictly consistent on the space of point measures and that the scores are normalised to $S\big(T(\delta_y), y\big) = 0$. Note that any score $S$ satisfying \ref{asm:eq:dirac} can be normalised to fulfil \ref{asm:eq:normalised}, by setting $\tilde{S}(z,y):= S(z,y) - S(T(\delta_y), y)$. Moreover, \cref{asm:score} leads to a unique characterisation of the families of scores of elicitable functionals, see e.g., Propositions \ref{prop:mean}, \ref{prop:quantile}, and \ref{prop:expectiles}.

\begin{definition}[Elicitability]\label{def:elicitable}
A functional $T \colon \bar{\M} \to \A$ is \textit{1-elicitable}
\footnote{
    A functional $T$ is called $k$-elicitable, $k \in \N$, if there exists a strictly $\M$-consistent scoring function $S\colon \A^k \times \R \to \R$, $\A^k \subseteq \R^k$ and a $\bz^* \in \R^{k-1}$, such that 
    $(\bz^*, T(F))= \argmin_{\bz \in \A^k} \int S(\bz, y)\d F(y)$, for all $F \in \M$. There are many functionals that are $k$-elicitable but not $1$-elicitable, see e.g., \cite{Fissler2016AS}.
    } 
on $\M\subseteq \bM$, if there exists a strictly $\M$-consistent scoring function $S$ for $T$.
Moreover, the functional $T$ has the following representation on $\M$
\begin{equation}\label{eq:elicitable-Bayes}
    T(F)= \argmin_{z \in \A} \int S(z, y)\,\d F(y)\,,
    \quad \forall\; F \in \M\,.
\end{equation}
\end{definition}
By Equation \eqref{eq:elicitable-Bayes}, a 1-elicitable functional is a Bayes act and a minimiser of an expected score \cite{Gneiting2011}. It is well known that the squared loss $S(z,y)=(z-y)^2$ elicits the mean. The squared Euclidean distance, however, is not the only strictly consistent score for the mean. Indeed, from \eqref{eq:elicitable-Bayes} we see that a 1-elicitable functional $T$ has infinitely many strictly consistent scores. In particular, the class of  strictly consistent scoring functions for the mean are given by the so-called Bregman divergences, which we recall next. 

\begin{definition}[Bregman divergence]
Let $\phi\colon \R \to \R$ be a convex function. The Bregman divergence associated with $\phi$ is defined as
\begin{equation*}
    B_\phi\big(z_1, z_2\big)
    := \phi(z_1) - \phi(z_2) - \phi'(z_2) (z_1-z_2)
    \,,\quad z_1,z_2\in\R\,,
\end{equation*}
where $\phi'(z):= \frac{d}{dz} \phi(z)$ denotes the derivative of $\phi$.
\label{Bregman}
\end{definition}
A Bregman divergence $B_\phi(z_1, z_2)$ can be seen as a measure of the deviation of $z_2$ from $z_1$. Note that in order to have a mathematical divergence (i.e., $B_\phi(z_1, z_2)=0$ if and only if $z_1=z_2$), $\phi$ needs to be strictly convex. While for the choice $\phi(z)=z^2$ the Bregman divergence coincides with the squared Euclidean distance, i.e., $B_{\phi}(z_1, z_2) = (z_1- z_2)^2$, in general, the Bregman divergence is not symmetric.

\subsection{Monge-Kantorovich divergences induced by scoring functions}

Next, we use scoring functions as cost functions in the Monge-Kantorovich optimal transport (OT) problem. In what follows, we call a function $c\colon \R^2\to \R_+$, $\R_+:=[0, \infty)$, that is lower-semi-continuous, a cost function. We recall the traditional Monge-Kantorovich optimisation problem and refer the reader to the book \cite{Santambrogio2015book} for further details.

\begin{definition}[Monge-Kantorovich optimal transport problem]
Let $c\colon \R^2 \to \R_+$ be a cost function. Then the Monge-Kantorovich optimisation problem with respect to the cdfs $F_1\in\Mo$ and $F_2\in\Mo$ is given by
\begin{equation}\label{eq:MK-opt}
    \inf_{\pi\in\Pi(F_1,\,F_2)} \;\left\{\,\int_{\R^2} c\big(z_1, z_2\big)\,\pi(\d z_1,\d z_2)\, \right\},
\end{equation}%
where $\Pi(F_1,F_2)$ denotes the set of all bivariate cdfs with marginal cdfs $F_1$ and $F_2$, respectively. A bivariate cdf that attains the infimum in \eqref{eq:MK-opt} is called an optimal coupling, which exists for any choice of cost function, see e.g., Theorem 1.7 in \cite{Santambrogio2015book}.
\end{definition}

For the cost function $c(z_1, z_2) := (z_1 - z_2)^p$, $p \ge 1$, we obtains the well-known $p$-Wasserstein distance 
\begin{align*}
    W_p (F_1, F_2)
    :&=    \inf_{\pi\in\Pi(F_1,\,F_2)} \;\left\{ \left(\,\int_{\R^2} (z_1 -  z_2)^p\,\pi(\d z_1,\d z_2)\, \right)^\frac{1}{p}\right\}
    \\
    &=
    \left(\int_0^1 \left|\Finv_1(u) - \Finv_2(u) \right|^p \d u\right)^{\frac1p}\,,
\end{align*}%
where $\Finv_i$ is the quantile function of $F_i$, $i = 1,2$, and where the last equality follows for cdfs on the real line, indicating that the comonotonic coupling is optimal \cite{dall1956SNS}.

This work introduces new asymmetric \emph{Monge-Kantorovich (MK) divergences} -- divergences on the space of cdfs -- that are defined by optimisation problem \eqref{eq:MK-opt}, where the cost functions are scoring functions. Thus, we not only introduce new MK divergences but also provide a novel perspective on scoring functions as OT cost functions.
Specifically, we consider cost functions $c(z_1, z_2): = S(z_2, z_1)$, for consistent scoring functions $S$. Note that the arguments of $c$ and $S$ are exchanged, this is due to different notational conventions in statistics and optimal transport, see e.g. Equation \eqref{eq:Bregman}. Our choice is justified in that we obtain the Bregman-Wasserstein divergence when choosing  any consistent scoring functions for the mean functional, see e.g., \cite{Rankin2023WP}.

We formulate the following definition. 
\begin{definition}[Monge-Kantorovich divergence]
Let $S$ be a $\M$-consistent scoring function for a functional $T$. Then the Monge-Kantorovich (MK) divergence induced by $S$ from the cdf $F_1\in \Mo$ to the cdf $F_2 \in \Mo$ is given by
\begin{equation}\label{eq:MK-divergence}
    \DS(F_1, F_2)
    :=
    \inf_{\pi\in\Pi(F_1,\,F_2)} \;\left\{\,\int_{\R^2} S\big(z_2, z_1\big)\,\pi(\d z_1,\d z_2)\, \right\}\,.
\end{equation}%
We term a bivariate cdf that attains the infimum in \eqref{eq:MK-divergence} an optimal coupling.
\end{definition}
The defined MK divergences may not necessarily be a divergence. They are non-negative and satisfy $\DS(F_1, F_1) = 0$, however, additional assumptions on the score $S$ are needed to ensure that $F_1=F_2$ implies $\DS(F_1, F_1)=0$, see e.g., \cite{Rankin2023WP} for the BW divergence. Furthermore, a MK divergence is in general not symmetric which is in contrast to, e.g., the $p$-Wasserstein distance. Clearly, the MK divergence depends on the choice of scoring function, however, for conciseness of the exposition, we refrain from writing $\DS_S$, whenever the scoring function is clear from the context.
The assumption to normalise the score, \cref{asm:score} \ref{asm:eq:normalised}, does not affect the optimal coupling as normalisation is achieved by subtracting from the score a function of $z_2$ only; see also the discussion after \cref{asm:score}.

If the cost function is the Bregman divergence, i.e., when in \eqref{eq:MK-opt} one considers $c(z_1, z_2) = B_\phi(z_1, z_2)$, we obtain the Bregman-Wasserstein (BW) divergence \cite{Rankin2023WP}
\begin{equation}\label{eq:BW}
    \DB_\phi (F_1, F_2)
    :=    \inf_{\pi\in\Pi(F_1,\,F_2)} \;\left\{ \,\int_{\R^2} B_\phi (z_1, \, z_2)\,\pi(\d z_1,\d z_2)\, \right\}\,,
\end{equation}
which reduces to the 2-Wasserstein distance for $B_\phi$ being the squared loss, i.e. $\phi(x) = x^2$.

\section{Optimal couplings for MK divergences}\label{sec:OT-maps}
This section is devoted to MK divergences induced by scores of elicitable risk functionals and their optimal couplings. Any elicitable risk functional, such as the mean, quantiles, expectiles, admits infinitely many consistent score, and thus gives raise to an infinite family of MK divergences. Furthermore, while the MK divergences are different for each score that elicits the risk functional, we show that the optimal couplings are the same.

\subsection{Bregman score as cost function}

The most commonly used 1-elicitable functional is the mean, whose family of scoring functions are the Bregman scores.

\begin{proposition}[Elicitability of Mean -- \cite{Gneiting2011}]\label{prop:mean}
    Let $\M$ denote the class of cdfs with finite mean. If $\phi$ is a (strictly) convex function with subgradient $\phi'$ and if 
$\int |\phi(y)|\,\mathrm{d}F(y) <\infty$ for all $F\in\M$, then the scoring function
\begin{equation}
    \label{eq:Bregman}
    S_{\phi}(z,y) := B_\phi(y,z)
    \,,\quad z,y\in\R\,,
\end{equation}
is (strictly) $\M$-consistent for the mean. Moreover, on the class of compactly supported measures, any (strictly) consistent score for the mean which is continuously differentiable in its first argument and which satisfies $S(y,y)=0$, is necessarily of the form \eqref{eq:Bregman}.
\end{proposition}

\begin{theorem}[Optimal coupling for BW-divergence]
\label{optimal-BW}
    The optimal coupling of any BW-divergence is the comonotonic coupling, i.e., $\big(\Finv_1(U), \Finv_2(U)\big)$, for any $U \sim U(0,1)$, is the optimal coupling.
\end{theorem}
In the language of OT, the comonotonic coupling implies that the \emph{optimal transport map} of \eqref{eq:BW}, i.e., the deterministic function mapping $F_1$ to $F_2$, is given by $\alpha(x):= \Finv_2\big(F_1(x)\big)$.
We refer the reader to \cite{Rankin2023WP}, who characterise the optimal transport map of the BW divergence for multivariate cdfs. Here, we provided for completeness a constructive proof that is valid for univariate cdfs.

\begin{proof}
Let $\phi$ be convex, then the BW-divergence becomes
\begin{align*}
    \DB_\phi(F_1, F_2)
    &= 
    \inf_{\pi\in\Pi(F_1,\,F_2)} \;\left\{ \,\int_{\R^2} \phi(z_1) - \phi(z_2) - \phi'(z_2) (z_1-z_2)\,\pi(\d z_1,\d z_2)\, \right\}
    \\
    &=
    \int_\R \phi(z_1)\d F_1(z_1) 
    +
    \int_\R\left(\phi'(z_2)z_2- \phi(z_2)  \right)\d F_2(z_2)
    \\
    & \quad 
    -
    \sup_{\pi\in\Pi(F_1,\,F_2)} \;\left\{ \,\int_{\R^2}  \phi'(z_2) z_1\,\pi(\d z_1,\d z_2)\, \right\}\,.
\end{align*}
Thus, to find the optimal coupling, we only have to solve the supremum. Rewriting in terms of random variables, we have 
\begin{equation}\label{eq:pf-BW-rv}
    \sup_{\pi\in\Pi(F_1,\,F_2)} \;\left\{ \,\int_{\R^2}  \phi'(z_2) z_1\,\pi(\d z_1,\d z_2)\, \right\}
    =
    \sup \; \,\E\left[  \phi'(Z_2) Z_1\, \right]\,,
    \quad Z_i \sim F_i\,, \;i = 1,2\,.
\end{equation}
where the supremum is over all copulae between $(Z_1, Z_2)$. Since $\phi'$ is increasing, it is well-known that the above supremum is attained by the comonotonic coupling; see e.g., Chapter 2 in \cite{ruschendorf2013mathematical}. Denoting the quantile function of $F_i$ by $\Finv_{i}$, $i = 1,2$, and since the quantile function of $\phi'(Z_2)$ is $\phi'(\Finv_2(\cdot))$, we obtain 
\begin{equation*}
    \sup \; \,\E\left[ Y Z_1\, \right]
    =
    \E\left[\phi'\left(\Finv_2 (U)\right)\, \Finv_1(U)\right]
    \,, \quad 
    U \sim U(0,1)\,,
    \end{equation*}
which concludes the proof.
\end{proof}

\subsection{Scores of generalised quantiles} 

Next, we consider quantiles and generalised quantiles. Quantiles are elicitable with the family of scoring functions called the generalised piecewise linear scores. 
 
\begin{proposition}[Elicitability of Quantile -- \cite{Gneiting2011}]\
\label{prop:quantile}
If $g$ is an increasing function, then the scoring function 
\begin{equation}
    \label{eq:score-quantile}
    S_g(z,y) 
    :=
    \big(\Id_{\{y\le z\}} - \alpha \big)\big(g(z) - g(y)\big)
    \,,\qquad z,y\in\R\,,
\end{equation}
is $\M$-consistent for the $\alpha$-quantile $\Finv(\alpha)$ if $\int |g(y)|\,\mathrm{d}F(y) <\infty$ for all $F\in\M$. 
If $g$ is strictly increasing and if for all $F\in\M$, $F(\VaR_\alpha(F)+\epsilon)>\alpha$ for all $\epsilon>0$, then \eqref{eq:score-quantile} is strictly $\M$-consistent.
Moreover, on the class of compactly supported measures, any consistent scoring function for $\Finv(\alpha)$ which is continuous in its first argument, which admits a continuous derivative for all $z\neq y$, and which satisfies $S(y,y)=0$, is necessarily of the form \eqref{eq:score-quantile}.
\end{proposition}

Hereafter we show that the comonotonic coupling is also optimal for the generalised piecewise linear score as cost function, i.e., when in \eqref{eq:MK-opt} we choose the cost function $c(z_1, z_2) :=  S_g(z_2,z_1)$.

\begin{theorem}[Optimal coupling for generalised piecewise linear scores]\label{thm:GPL}
The optimal coupling of the MK minimisation problem induced by any consistent generalised piecewise linear score is the comonotonic coupling.
\end{theorem}
\begin{proof}
Let $g$ be increasing, then the MK divergence induced by the score \eqref{eq:score-quantile} is
\begin{align*}
    \DS(F_1, F_2)
    &=
    \inf_{\pi\in\Pi(F_1,\,F_2)} \;\left\{\,\int_{\R^2} \big(\Id_{\{z_1\le z_2\}} - \alpha \big)\big(g(z_2) - g(z_1)\big)\,\pi(\d z_1,\d z_2)\, \right\}
    \\
    &=
     \alpha \int_\R g(z_1) \d F_1(z_1) - \alpha \int_\R g(z_2) \d F_2(z_2) 
    \\
    & \quad +
    \inf_{\pi\in\Pi(F_1,\,F_2)} \;\left\{\,\int_{\R^2} \Id_{\{z_1\le z_2\}} \big(g(z_2) - g(z_1)\big)\,\pi(\d z_1,\d z_2)\, \right\}\,.
\end{align*}
The infimum rewritten in terms of random variables, is
\begin{align*}
    \inf_{\pi\in\Pi(F_1,\,F_2)} \;\left\{\,\int_{\R^2}  \big(g(z_2) - g(z_1)\big)_+\,\pi(\d z_1,\d z_2)\, \right\}
    =
    \inf \;
    \E\left[ \, \big(g(Z_2) - g(Z_1)\big)_+\,
    \right]\,,
\end{align*}
where the infimum is over all copulae of $(Z_1, Z_2)$. Note that for any bivariate random vector $(X,Y)$, $X\sim F_X$, $Y\sim F_Y$, it holds that \cite{meilijson1979convex}
\begin{equation}
    X^a + Y^a \prec_{cx} 
    X +Y 
    \prec_{cx} X^c + Y^c\,,
\label{FHbounds2}
\end{equation}
where $\prec_{cx}$ denotes inequality in convex order and where $(X^a,Y^a)$ and 
$(X^c,Y^c)$ denote the antitonic resp. comonotonic pair with marginal distributions $F_X$ and $F_Y$, that is $(X^a,Y^a):=(F_X^{-1}(U),F_Y^{-1}(1-U))$ and $(X^c,Y^c)=(F_X^{-1}(U),F_Y^{-1}(U))$, $U\sim U(0,1)$. As $g$ is increasing, $(g(Z_2^c), - g(Z_1^c))$ is an antitonic pair and the first inequality in \eqref{FHbounds2} implies that 
\begin{equation*}
    g(Z_2^c) - g(Z_1^c) \prec_{cx} g(Z_2) - g(Z_1).
\end{equation*}
Since $f(x) = x_+$ is a convex function we obtain that 
\begin{equation*}
    \inf \;
    \E\left[ \, \big(g(Z_2) - g(Z_1)\big)_+\,
    \right]
    =
    \E\left[ \, \left(g\big(\Finv_2(U)\big) - g\big(\Finv_1(U)\big)\right)_+\,
    \right]\,,
    \quad U \sim U(0,1)\,,
\end{equation*}
which concludes the proof.
\end{proof}

There exist many generalisations of quantiles, such as $L^p$ quantiles, $M$ quantiles and $\Lambda$-quantiles. We show in the sequel that consistent scores of these generalised quantiles give raise to the comonotonic 
coupling being optimal for \eqref{eq:MK-divergence}. 
In this section, we discuss the $\Lambda$-quantile, while the $L^p$ and $M$ quantiles are considered in \Cref{sec:Decom-score}.

For a monotone and right-continuous function $\Lambda \colon \R \to [\underline{\lambda}, \, \overline{\lambda}]$, $0< \underline{\lambda}<\overline{\lambda}<1$, the $\Lambda$-quantile is defined by \cite{Burzoni2017QF}
\begin{equation}\label{eq:Lambda-quant}
    T_\Lambda(F):= 
    \inf\big\{
    y \in \R ~:~ F(y) > \Lambda(y)
    \big\}\,.
\end{equation}
The $\Lambda$-quantile is elicitable with score
\begin{equation}\label{eq:Lambda-score2}
    S(z,y)
    =
    (z-y)_+ - \int_y^z \Lambda(s) \, \d s\,,
\end{equation}
on the space of cdfs that admit only one unique crossing point with $\Lambda(\cdot)$.

\begin{theorem}[Optimal coupling for $\Lambda$-quantile score]\label{thm:lambda}
The optimal coupling of the MK minimisation problem induced by the score given in \eqref{eq:Lambda-score2} is the comonotonic coupling.
\end{theorem}
\begin{proof}
As         $\partial_z S(z,y)
        =
        \Id_{\{ y \le z\}} - \Lambda(z)$
    is decreasing in $y$, for all $z$, and $\partial_y S(z,y)
       =
       \Lambda(y) - \Id_{\{ y \le z\}} $ is decreasing in $z,$  for all $y$, the score \eqref{eq:Lambda-score2} is submodular (see also chapter 6.D in \cite{marshall1979inequalities}).
 Hence, since the comonotonic coupling is optimal for submodular cost functions \cite{rachev1998mass}, the result follows.  
\end{proof}

\subsection{Expectile score} 
The comonotonic coupling is also optimal when the cost function is a score that elicits the $\alpha$-expectile.  

\begin{proposition}[Elicitability of Expectiles -- \cite{Gneiting2011}]\label{prop:expectiles}
    Let $\M$ denote the class of cdfs with finite mean. If $\phi$ is (strictly) convex with subgradient $\phi'$ and if 
$\int |\phi(y)|\,\mathrm{d}F(y) <\infty$ as well as $\int |y|\,\mathrm{d}F(y) <\infty$ for all $F\in\M$, then the scoring function 
\begin{equation}
    \label{eq:expectile}
    S(z,y) = \big|\Id_{\{y\le z\}} - \alpha \big|\, B_\phi(y,z)
    \,,\quad z,y\in\R\,,
\end{equation}
is (strictly) $\M$-consistent for the expectile\footnote{The expectile was first introduced in \cite{newey1987asymmetric} as the  minimiser of \eqref{eq:elicitable-Bayes} for the scoring function \eqref{eq:expectile} with $B_\phi(y,z)=(z-y)^2.$  }.
Moreover, on the class of compactly supported measures, any (strictly) consistent score for the expectile which is continuously differentiable in its first argument and which satisfies $S(y,y)=0$, is necessarily of the form \eqref{eq:Bregman}.
\end{proposition}

\begin{theorem}[Optimal coupling for $\alpha$-expectile scores]\label{thm:expectile}
The optimal coupling of the MK minimisation problem induced by any consistent score given in $\eqref{eq:expectile}$ is the comonotonic coupling.
\end{theorem}

\begin{proof}
The MK divergence induced by the scoring function \eqref{eq:expectile} 
is given as
\begin{align*}
    \DS(F_1, F_2)
    &=
    \inf_{\pi\in\Pi(F_1,\,F_2)} \;\left\{\,\int_{\R^2} \big|\Id_{\{z_1\le z_2\}} - \alpha \big| \, B_\phi(z_1,z_2)\,\pi(\d z_1,\d z_2)\, \right\}
    \\
    &=
     \inf_{\pi\in\Pi(F_1,\,F_2)} \; \bigg\{\int_{\R^2} (1-\alpha) B_\phi(z_1,z_2) \Id_{\{z_1 \leq z_2\}} \pi(\d z_1,\d z_2)\  
    \\
    & \hspace{75pt} +
      \int_{\R^2} \alpha B_\phi(z_1,z_2) \Id_{\{z_1 > z_2\}} \pi(\d z_1,\d z_2) \bigg\}\,.
\end{align*}
Convexity of  $\phi$ implies that $B_\phi(z_1,z_2)$ is submodular. It readily verifies that
$$(1-\alpha) \partial_{z_1} B_\phi(z_1,z_2) \Id_{\{z_1 \leq z_2\}} + \alpha \partial_{z_1} B_\phi(z_1,z_2) \Id_{\{z_1 > z_2\}}$$ is decreasing in $z_2$ for all $z_1$, which implies that the function
$$(1-\alpha) B_\phi(z_1,z_2) \Id_{\{z_1 \leq z_2\}} + \alpha B_\phi(z_1,z_2) \Id_{\{z_1 > z_2\}}$$ 
is submodular \cite{marshall1979inequalities}. As the comonotonic coupling is optimal for submodular cost functions, the result follows.  
\end{proof}

\subsection{Shortfall score}
A popular class of risk measures is the family of shortfall risk measures. We first recall their definition and the scoring functions that elicit shortfall risk measures.

\begin{definition}[Shortfall risk measure]\label{def:score-shortfall}
Let $\ell\colon \R\to\R$ be an increasing and non-constant function satisfying $\ell(w)<0$, whenever $w <0$, and $\ell(w)>0$, whenever $w >0$. Then the shortfall risk measure $T_\ell$ is defined by
\begin{equation}
    T_\ell(F)
    :=
    \inf\left\{ x \in \R ~\Big|~ \int \ell( w-x) \, \d F(w) \le 0\right\}\,,
\end{equation}
whenever the infimum exists.
If furthermore $\ell $ is left-continuous and strictly increasing on either $(-\infty, \ep)$ or $(\ep, +\infty)$, for $\ep>0$, then $T_\ell$ is 1-elicitable with strictly consistent score \cite{BelliniQF2015}
\begin{equation}\label{eq:shortfall}
    S(z,y) = \int_0^{y-z} \ell(s) \, \d s\,.
\end{equation}
\end{definition}

\begin{theorem}[Optimal coupling for shortfall scores]\label{thm:shorfall}
    The optimal coupling of the MK minimisation problem induced by the scoring function given in \eqref{eq:shortfall} is the comonotonic coupling.
\end{theorem}
\begin{proof}
As $\partial_z S(z,y)   = -\ell(y-z) $
    is decreasing in $y$, for all $z$, the scoring function in \eqref{eq:shortfall} is submodular. 
\end{proof}

\subsection{Decomposable scores}\label{sec:Decom-score}
Next we study a family of scores that elicits different risk functionals including quantiles, expectiles, and $M$ and $L^p$ quantiles.

\begin{theorem}[Optimal coupling for decomposable scores]
Assume the scoring function is of the form 
\begin{equation}\label{eq:decom-score}
    S(z,y)= \phi(|z-y|) \left( \alpha \,\Id_{\{y > z\}}
    +
    \beta \,\Id_{\{y \le z\}}\right)\,,
\end{equation} 
where $\alpha, \beta \in [0,1]$ and $\phi \colon \R_+ \to \R_+$ is an increasing convex function satisfying $\phi(0) =0$. Then the optimal coupling of the MK minimisation problem induced by any score given in \eqref{eq:decom-score} is the comonotonic coupling.
\end{theorem}

\begin{proof}
    We rewrite the scoring function to
    \begin{equation*}
    S(z,y)
    = \phi\big(|z-y|\big) \left( \alpha \,\Id_{\{y > z\}}
    +
    \beta \,\Id_{\{y \le z\}}\right)
    = 
    \alpha \, \phi\big((y-z)_+\big)
     + \beta \, \phi\big((z-y)_+\big)\,.
    \end{equation*}
    As the function $\phi(x_+)$ is convex, the results follows using similar arguments as in the proof of \Cref{thm:GPL}.
\end{proof}

The scoring function \eqref{eq:decom-score} elicits the $\alpha$-quantile for the choice $\phi(x) = x$ and $\beta = 1-\alpha$ and the $\alpha$-expectile with $\phi(x) = x^2$ and $\beta = 1-\alpha$. More generally,  the score \eqref{eq:decom-score} for any arbitrary convex $\phi$ and $\beta=1-\alpha$ elicit the so-called $M$-quantiles, defined via \eqref{eq:elicitable-Bayes} whenever they exits, see e.g., \cite{breckling1988}. $M$-quantiles subsume $L^p$-quantiles, which correspond to $\phi(x) = x^p$, $p \in [1,\infty)$, see
\cite{chen1996}. 

\subsection{Osband's transformation of scores}
Osband's principle is highly regarded in statistics to create new elicitable functionals \cite{Gneiting2011}. Here, we use Osband's principle to characterise the optimal coupling induced by consistent scores for monotone transformations of elicitable functionals.

\begin{proposition}[Osband's principle for OT]\label{prop:Osband}
Let $T$ be a strictly monotone transformation of a 1-elicitable functional $\tilde{T}$, i.e. $T:= g \circ \tilde{T}$ for $g\colon \R \to \R$ strictly monotone, and denote by $\tilde{S}$ a (strictly) $\M$-consistent scoring function for $\tilde{T}$. Then $T$ is 1-elicitable with (strictly) $\M$-consistent scoring function 
\begin{equation}\label{eq:osband-score}
    S(z,y):= \tilde{S}\big(g^{-1}(z), y\big)\,.
\end{equation}
Let the corresponding MK minimisation problem of $\tilde{S}$ be attained by the comonotonic coupling. 
If further $g$ is increasing (decreasing) then the optimal coupling of the MK minimisation problem induced by the score \eqref{eq:osband-score} is the comonotonic (antitonic) coupling.
\end{proposition}

\begin{proof}
By Osband's principle, the functional $T:= g \circ \tilde{T}$ for any bijective function $g\colon \R \to \R$ is 1-elicitable with scoring function given in \eqref{eq:osband-score}, see e.g., Theorem 4 in \cite{Gneiting2011}. If $g$ is strictly increasing then, $\frac{d}{dx}g^{-1}(x) = \big\{g'\, \big(g^{-1}(x)\big)\big\}^{-1} > 0$, where $g'(x):= \frac{d}{dx}g(x)$. Thus, the optimal coupling is comonotonic. If $g$ is strictly decreasing, then $\frac{d}{dx}g^{-1}(x) <0$ and the antitonic coupling is optimal.
\end{proof}

Osband's principle can be used to create new MK divergences. Indeed any strictly monotonic transformation of an elicitable risk functional leads to a new MK divergence, where the optimal coupling follows from \cref{prop:Osband}. Here, we give an example of the reciprocate of risk functionals. 

\begin{corollary}[Inverse of risk functionals]
Let $\tilde{T}$ be a strictly positive 1-elicitable risk functional and assume that the corresponding MK minimisation problem is attained by the comonotonic coupling. 

Next consider the functional
\begin{equation*}
    T(F):= \frac{1}{\tilde{T}(F)}\,,
\end{equation*}
then the MK minimisation problem induced by any consistent score for $T$ is attained by the antitonic coupling.
\end{corollary}
\begin{proof}
This follows by \Cref{prop:Osband} with $T(\cdot) = g\big(\tilde{T}(\cdot)\big)$,
where $g(x) = \frac{1}{x}$ is decreasing.
\end{proof}

The next application of Osband's principle concern transformations of the distribution. For this, let $F_Z$ denote the distribution of a random variable $Z$.

\begin{proposition}[Osband's principle for OT]
\label{prop:Osband-Y}
    Let $\tilde{T}$ be a 1-elicitable risk functional. For a function $h\colon \R \to \R$, consider the risk functional $T(F_Y):= \tilde{T}\left(F_{h(Y)}\right)$. 
    Assume that the MK minimisation problem induced by a consistent score $\tilde{S}$ for $\tilde{T}$ is attained by the comonotonic coupling. If $h$ is increasing (decreasing), then the optimal coupling of the MK minimisation problem induced by a consistent score of $T$ is the comonotonic (antitonic) coupling.
\end{proposition}
\begin{proof}
    A score $\tilde{S}$ is strictly consistent for $\tilde{T}$ applied to $F_{h(Y)}$ if and only if it holds for all $z \neq \tilde{T}(F_{h(Y)})$ that
    \begin{equation}\label{eq:pf-score-entropic}
    \int \tilde{S} \Big(\tilde{T}\big(F_{h(Y)}\big),\tilde{y}\Big)\,\d F_{h(Y)}(\tilde{y}) 
        \;\le\;
    \int \tilde{S}\left(z,\tilde{y}\right)\,\d F_{h( Y)}(\tilde{y})\,.
    \end{equation}
The inequalities \eqref{eq:pf-score-entropic} are equivalent to 
    \begin{equation}
    \int \tilde{S}\Big(T\big(F_{Y}\big), \, h(y)\Big)\,\d F_{Y}( y) 
        \le
    \int \tilde{S}\big(z,\,  h(y)\big)\,\d F_Y( y)\,,
    \end{equation}
for all $z\neq T(F_Y)$. Thus, the score $S(z,y):= \tilde{S}\big(z, h(y)\big)$ is a consistent score for $T(F_Y)$. The reminder follows using similar arguments as in the proof of \cref{prop:Osband}.
\end{proof}

An immediate example of \cref{prop:Osband-Y} is that of the functional $T(F_Y) = \E\left[\frac{1}{Y}\right]$ is elicitable and induces MK divergences for which the antitonic coupling is simultaneously optimal for any consistent score of $T$. 

Using Osband's principles for OT, we derive the optimal coupling for the scores of entropic risk measures. The entropic risk measure with parameter $\gamma >0$, also known as the exponential premium principle in actuarial science \cite{Gerber1974ASTIN}, is defined for $F\in\Mo$, by 
\begin{equation*}
    T^\gamma(F) 
    :=
    \frac{1}{\gamma}\, \log \, \int e^{\gamma y} \, \d F(y)\,, 
\end{equation*}
Any (strictly) $\M$-consistent scoring function (under mild regularity conditions) for the entropic risk measure satisfying $S(y,y) = 0$ is given by
\begin{equation*}
    S(z,y)
    =
    \phi\big(e^{\gamma y}\big)
    -
    \phi\big(e^{\gamma z}\big)
    +  
    \phi'\big(e^{\gamma y}\big) \left(e^{\gamma z}-e^{\gamma y}\right)\,,
\end{equation*}
where $\phi \colon \R \to \R$ is (strictly) convex and $\int|\phi\big(e^{\gamma y}\big)|\, \d F(y) < \infty$ for all $F \in \M$; for more detail see Appendix A1 in \cite{Fissler2023EJOR}.

\begin{corollary}[Optimal coupling for entropic scores]
The optimal coupling of the MK minimisation problem induced by any consistent scoring function for the entropic risk measure is the comonotonic coupling.    
\end{corollary}
\begin{proof}
Denote by $\tilde{T}(G) := \int x\,  \d G(x) $ the expectation functional with consistent score $S_\phi$, the Bregman score, i.e. \eqref{eq:Bregman}. Then, it holds that 
\begin{equation*}
    T^\gamma(F) := g\circ \tilde{T}\big(F_{h(Y)}\big)\,,
\end{equation*}
where $g(x):= \frac{1}{\gamma} \log(x)$ and $h(x):= e^{\gamma x}$. By Osband's principles, see the proofs of Propositions \ref{prop:Osband} and \ref{prop:Osband-Y}, the consistent score for $T^\gamma$ is thus given by 
\begin{equation*}
    S(z,y) 
    = S_\phi \left(g^{-1}(z), \, h( y)\right)
    =
    \phi(e^{\gamma z}) - \phi(e^{\gamma y}) - \phi'(e^{\gamma y}) (e^{\gamma z}-e^{\gamma y})\,,    
\end{equation*}
which is indeed the strictly consistent scoring function for the entropic risk measure. Finally, applying Propositions \ref{prop:Osband} and \ref{prop:Osband-Y} concludes the proof.
 \end{proof}

\subsection{Scores of law-invariant risk measures}
In this section we consider MK divergences derived from elicitable convex and coherent risk measures. We show that the optimal coupling of the MK divergence induced by any strictly consistent score that elicits a coherent risk measure is the comonotonic coupling. For this, denote by $T  \colon \M^0 \to \R$ a law-invariant risk measure. A law-invariant risk measure can equivalently be defined on the space of random variables $\L$, as the functional $T\colon \L \to \R$, by setting $T[X] := T(F_X)$, whenever $X$ has cdf $F_X$. We use the notation $T(\cdot)$, when the risk measure is viewed as a function of cdfs and $T[\cdot]$, when applied to random variables.
We say a law-invariant risk measure $T$ is 
\begin{enumerate}[label = $(\roman*)$]
    \item \textbf{monotone:} if $T[X] \le T[Y]$ whenever $X \le Y $ $\P$-a.s., $X,Y \in \L^\infty$,
    \label{it:monotone}

    \item \textbf{translation invariant:} if $T[X + m] = T[X] + m$, for all $X\in \L^\infty$ and $m\in \R$,
   \label{it:trans-inv}

    \item \textbf{positive homogeneous:} if $T[\lambda \,X] = \lambda \,T[X]$, for all $X\in \L^\infty$ and $\lambda \ge 0$,
    \label{it:pos-hom}
    
    \item \textbf{convex:} 
    \label{it:convex}
    if $T\big[ \lambda\, X + (1-\lambda) \,Y] \le \lambda \,T[X] + (1-\lambda) \,T[Y]$, for all $X,Y \in \L^\infty$ and $\lambda \in [0,1]$,

\end{enumerate}
A law-invariant functional $T$ is called a convex risk measure, if it satisfies the properties \ref{it:monotone}, \ref{it:trans-inv}, and \ref{it:convex}, and a coherent risk measure if it additionally fulfils \ref{it:pos-hom}. For discussions and interpretation of these properties, we refer the reader to \cite{Follmer2002book} and reference therein.

\begin{theorem}[Coherent Risk Measures]
Let $T$ be a 1-elicitable coherent risk measure satisfying $T(0)=0$, and let $S$ be any strictly consistent score for $T$. Then, the optimal coupling of the MK minimisation problem induced by the score $S$ is the comonotonic coupling. 
\end{theorem}

\begin{proof}
The class of 1-elicitable coherent risk measures coincides with the $\alpha$-expectiles where $\alpha \in [0.5,1]$, see Corollary 12 in \cite{Steinwart2014CLT}. Applying \Cref{thm:expectile} concludes the proof.
\end{proof}

Next, we consider convex risk measures, which require additional assumptions on the score.

\begin{definition}\label{def:score-prop}
We consider the following additional properties on scoring functions
\begin{enumerate}[label = $(\roman*)$]
    \item $S(z,y) =0$ if and only if $z = y$
    
    \item  $S(z,y)$ is continuous in both $z$ and $y$,
    
    \item \label{def:score-accuracy}
    $S(z,y)$ is increasing in $z$, for $z > y$ and decreasing in in $z$ for $z< y$
    
    \item for all $z $ in a neighbourhood of 0, $S(z,y) \le \psi(y)$, where $\psi \colon \R \to [1, + \infty)$ is a continuous gauge function.
\end{enumerate}
\end{definition}

Property \ref{def:score-accuracy} is called accuracy rewarding in the statistical literature  \cite{Lambert2008ConfEC}, since it implies that if $T(F)< z_1<z_2$ or $T(F)>z_1 > z_2$, then 
\begin{equation*}
    \int S\big(T(F), \, y \big)\, \d F(y)
    <
    \int S\big(z_1, \, y \big)\, \d F(y)
    <
    \int S\big(z_2, \, y \big)\, \d F(y)\,.
\end{equation*}
Thus, for accuracy rewarding scores the further away the estimates $z$ are from the truth $T(F)$, the larger are their expected scores.

\begin{theorem}[Convex Risk Measures]
Let $T  \colon \M^\infty \to \R$, where $\M^\infty \subseteq \M^0$ is the space of all cdfs with bounded support, be a law-invariant and 1-elicitable convex risk measure. Then, there exists a strictly $\M^\infty$-consistent score for $T$ that satisfies the properties in \Cref{def:score-prop}, and for which the optimal coupling of the MK minimisation problem is the comonotonic coupling. 
\end{theorem}
\begin{proof}
The class of convex risk measures that are  1-elicitable coincides with the class of shortfall risk measures $T_\ell$ with convex loss function $\ell$, see Theorem 4.6 in \cite{BelliniQF2015}. Applying \Cref{thm:shorfall} concludes the proof.
\end{proof}

\section{Applications}\label{sec:app}
\subsection{Worst-case distortion risk measures}
Worst-case distortion risk measures are often used in robust stochastic optimisation, see e.g., \cite{bernard2022robust} and \cite{Pesenti2023SIAM}. We call $g \colon [0,1] \to [0,1]$  a distortion function, if it is non-decreasing and satisfies $g(0) = 0$ and $g(1) = 1$. A distortion risk measure evaluated at a cdf $G$ with quantile function $\Ginv$ is defined as the Choquet integral
\begin{equation*}
H_g(\Ginv) 
	:= - \int_{-\infty}^0 \left( 1 - g(1 -  G(x))\right) \,\mathrm{d}x + \int_0^{+ \infty}g\big(1 - G(x)\big)\,\mathrm{d}x\,,
\end{equation*}
in which at least one of the two integrals is finite. If $g$ is absolutely continuous, then the distortion risk measure $H_g$ has representation 
\begin{align}
H_g(\Ginv) &= \int_0^1 \gamma(u)\Ginv (u) \,\mathrm{d}u\,, \label{eq: distortion Choquet integral}
\end{align}
with weight function $\gamma(u) := \partial_- g(x)|_{x = 1 - u}, ~ 0 < u < 1$, which satisfies $\int_0^1 \gamma(u)\mathrm{d}u = 1$ and where $\partial_-$ denotes the derivative from the left. In the sequel we assume that representation \eqref{eq: distortion Choquet integral} holds. The class of distortion risk measures is broad and contains the majority of risk measures used in financial risk management practice including the quantile (Value-at-Risk) and the average of upper quantiles (Tail Value-at-Risk).

 Let $\phi$ be a convex function and $F$ be a given reference cdf satisfying $H_g(\Finv)< \infty$, and consider the optimisation problem 
 \begin{equation}
 \label{opt:WC-RMtoto}
     \max_{\Ginv \in\Minv} \; H_g\big(\Ginv\big)
 \,,
     \quad s.t. \quad
     \DB_\phi\big(G, \, F\big) \,\le \,\ep\,,
 \end{equation}
where $\Minv$ is the set of left-continuous quantile functions on $\R$. 
Optimisation problems of the form  \eqref{opt:WC-RMtoto} are of  great interest in financial risk management, as risk measures and the values thereof drive decision making (e.g., through solvency requirements), but typically are subject to distributional uncertainty, here quantified via the BW divergence. Specifically, in the optimisation problem \eqref{opt:WC-RMtoto} one aims to determine the worst-case value a distortion risk measure can attain when the underlying cdf belongs to a set of cdfs that are close, in the BW divergence, to a reference cdf $F$. 
The special case of optimisation problem \eqref{opt:WC-RMtoto} for $\phi(x)=x^2$, i.e. when the BW divergence reduces to the 2-Wasserstein distance, was recently solved in \cite{bernard2022robust}. Here, we significantly extend this result.

\begin{theorem}[Worst-case Distortion Risk Measures]\label{thm:WC-RM-BW}
Assume that the distortion function $g$ is strictly concave and that $\phi$ is strictly convex. If there exists a solution to the optimisation problem \eqref{opt:WC-RMtoto}, then it is uniquely given by
\begin{equation}
\label{eq:solution}
    \Ginv_{\lambda^*}(u)
    := 
    \left(\phi'\right)^{-1}
    \left(
    \phi'\big(\Finv(u)\big) + \frac{1}{\lambda^*}\gamma(u)
    \right)
    \,,
\end{equation}
where $\lambda^* > 0$ is the unique solution to $\DB_\phi\big(G_{\lambda}, \, F\big) = \ep$. 
\end{theorem}

\begin{proof}
A solution to problem \eqref{opt:WC-RMtoto} must attain a BW divergence $\ep_0$, $0\leq \ep_0 \leq \ep$, and thus solves for all $\lambda >0,$ the optimisation problem  
\begin{align*}
   \argmax_{\Ginv \in \Minv} \;
    \int_0^1  \Big(\Ginv(u)\gamma(u) - \lambda \big(\phi\big(\Ginv(u)\big)-\phi\big(\Finv(u)\big)
    -\phi'\big(\Finv(u)\big)\Ginv(u) \ + \
     \phi'\big(\Finv(u)\big)\Finv(u)
    - \ep_0\big)\Big)\mathrm{d}u.
\end{align*}
Equivalently, the solution solves for all $\lambda > 0$  the optimisation problem 
\begin{equation*}
    \argmin_{\Ginv \in \Minv} \;  
        \int_0^1 
     \big(\phi\big(\Ginv(u)\big)
    - k_\lambda(u)\Ginv(u)\big)\mathrm{d}u\,,
\end{equation*}
where $k_\lambda(u):= \phi'\big(\Finv(u)\big)
     + \frac{1}{\lambda} \gamma(u).$
    Note that for any given function  $k_{\lambda}(\cdot), \lambda>0$, 
    the function $\phi\big(\Ginv(u)\big)- k_\lambda(u) \Ginv(u)$ is convex in $\Ginv(u)$. By direct optimisation and noting that $k_\lambda(u)$ is an increasing function, we thus obtain that a solution to \eqref{opt:WC-RMtoto} is of the form 
\begin{equation}
\label{eq:help}
    \Ginv_{\lambda}(u)
    := 
    \left(\phi'\right)^{-1}
    \left(
    \phi'\big(\Finv(u)\big) + \frac{1}{\lambda}\gamma(u)
    \right)
    \,,
\end{equation}
for some $\lambda >0$.
Furthermore, $\lambda_1 < \lambda_2$ implies that 
    $\Ginv_{\lambda_1}(u) > \Ginv_{\lambda_2}(u)$. As a consequence,  it also holds that $H_g(\Ginv_{\lambda_1}) > H_g(\Ginv_{\lambda_2})$ and moreover that (recall that $\phi$ is convex) $\DB_\phi( G_{\lambda_1},F) > \DB_\phi(G_{\lambda_2},F)$. Hence, for a solution to be optimal, $\lambda^*$ must satisfy $\DB_\phi(G_{\lambda^*},F)=\ep$. The existence of $\lambda^*$ follows since $\DB_\phi( G_{\lambda},F)$ is continuously decreasing in $\lambda >0$,  $\lim_{\lambda \to 0 } {\DB_\phi( G_{\lambda},F)} = \infty$, and $\lim_{\lambda \to \infty } {\DB_\phi( G_{\lambda},F)} = 0$. Moreover, $\lambda^*$ is unique. 
\end{proof}

\subsection{Cheapest payoffs}
Let $S_{T}$ represents the non-negative random value of a risky asset at time $T>0.$ 
Further consider a bank account earning the continuously
compounded risk-free interest rate $r\in\mathbb{R}$. 
We define the set of payoffs
\begin{equation*}
    \mathcal{X}:=\Big\{ g(S_{T}) ~|~\,g:\R_+\to\R_+\text{ is measurable and }  \E_{\mathbb{Q}}\big[|g(S_{T})|\big]<\infty    \Big\}\,,
\end{equation*}
where $\Q$ is a pricing
measure equivalent to $\P$, and $\E_\Q[\cdot]$ denotes the expectation under $\Q$. For any payoff $X\in\mathcal{X}$, its initial cost is 
\begin{equation}
\label{cost}
  c(X):=e^{-rT}\E_{\Q}[X]
  =
  \E_\P[\xi X]\,,
\end{equation}
where the last equality follows by defining the state-price-density $\xi :=e^{-rT}\frac{d \Q}{d\P}$. We assume that the state-price density $\xi$ is continuously distributed and denote its cdf under $\P$ by $F_\xi$.
In the sequel, all cdfs of random variables are taken with respect to $\mathbb{P}$. 

Let $\phi$ be a convex function and $F$ be a cdf $F$ satisfying $\E_\P[\Finv_\xi(U)\Finv(1-U)]< \infty$, for a uniform random variable $U$. We consider the problem of finding the cheapest payoff under the constraint that its cdf is close to a benchmark. Specifically, we require that its cdf lies withing a BW-ball around a benchmark cdf $F$, that is we consider the optimisation problem
\begin{equation} \label{eq:cost_eff_sproblem_class2}      \min_{X\in\mathcal{X}} \; c(X)
 \,,
      \quad s.t. \quad
     \DB_\phi\big(F_X, \, F\big) \,\le \,\ep\,.
\end{equation}
A special case of \eqref{eq:cost_eff_sproblem_class2} was  first considered in \cite{dybvig1988inefficient} who, instead of the BW divergence, consider payoffs with fixed cdf $F$. That is, they solve
\begin{equation}
 \label{opt:WC-RMtoto3}
     \min_{X\in\mathcal{X}} \; c(X)
 \,,
     \quad s.t. \quad
     F_{X}\equiv F\,,
 \end{equation}
and show that the unique solution to optimisation problem \eqref{opt:WC-RMtoto3} is given by 
$X^*:=\Finv \big(1-F_{\xi}(\xi)\big)$.    
The payoff $X^*$ is called cost-efficient as it is decreasing in $\xi$, and thus yields the cheapest payoff with given cdf.  It is well understood that an investor who aims to maximize an increasing law-invariant objective (e.g., an expected utility) necessarily purchases a cost-efficient payoff (\cite{Follmer2002book}, \cite{carlier2006law}, \cite{ruschendorf2020construction}). This has led to the quantile formulation of optimal payoff choice selection; see e.g., \cite{he2011portfolio}, \cite{bernard2015optimal}, \cite{ruschendorf2020construction}, \cite{boudt2022optimal}, and \cite{Pesenti2023SIAM} for studies and applications. Next, we provide the solution to optimisation problem \eqref{eq:cost_eff_sproblem_class2}.

\begin{theorem}[Cheapest payoffs]
The solution to  optimisation problem \eqref{eq:cost_eff_sproblem_class2} is given 
by 
\begin{equation}
X^*:={\Ginv_{\lambda^*}}\big(1-F_{\xi}(\xi)\big)\,,
\label{cheap}
\end{equation}
in which
\begin{equation}
\label{eq:solution2}
    \Ginv_{\lambda^*}(u)
    := 
    \left(\phi'\right)^{-1}
    \left(
    \phi'\big(\Finv(u)\big) - \frac{1}{\lambda^*}\,{\Fxiinv}(1-u)
    \right)
    \,,
\end{equation}
where $\lambda^* > 0$ is the unique solution to $\DB_\phi\big(G_{\lambda^*}, \, F\big) = \ep$. \end{theorem}

\begin{proof}
We first show that any solution has to be cost-efficient. To this regard, let $X$ be a solution to  optimisation problem \eqref{eq:cost_eff_sproblem_class2} with cdf $G$. Define the random variable $Y:=\Ginv\big(1-F_\xi (\xi)\big)$ which has cdf $G$, and thus its cdf lies within the BW-ball. Moreover, unless $Y=X$, $\P$-a.s., $Y$ is strictly cheaper than $X$. Hence, an  optimal solution has to be cost-efficient. Furthermore, from \eqref{cost} it follows that the cost of any cost-efficient payoff $Y$ with cdf $G$ can be expressed as 
\begin{equation}
\label{CEcost}
c(Y)=-\int_0^1 \gamma(u)\, \Ginv(u)\, \mathrm{d}u\,,
\end{equation}
where $\gamma(u):=-{\Fxiinv}(1-u)$ is an increasing function. We observe that solving optimisation problem  \eqref{eq:cost_eff_sproblem_class2} is equivalent to solving an optimisation problem of the form \eqref{opt:WC-RMtoto}. Next, we apply \Cref{thm:WC-RM-BW}, as an inspection of its proof reveals that \Cref{thm:WC-RM-BW} also holds when the (increasing) weighting function $\gamma$ is negative. 
\end{proof}

\bibliographystyle{siamplain}
\bibliography{references}

\end{document}